\documentclass[pdflatex,sn-mathphys]{sn-jnl}% Math and Physical Sciences Reference Style
\jyear{2021}%

\theoremstyle{thmstyleone}%
\newtheorem{theorem}{Theorem}%  meant for continuous numbers
\newtheorem{proposition}{Proposition}% to get separate numbers for theorem and proposition etc.

\usepackage{mathtools}
\usepackage{algorithm}
\usepackage{eqnarray,amsmath}
\usepackage{amssymb}
\usepackage{tabularx}
\usepackage{float}
\usepackage{multirow}
\usepackage{ragged2e}
\usepackage{algpseudocode}

\newtheorem{corollary}{Corollary}
\newtheorem{lemma}{Lemma}

\newtheorem{remark}{Remark}
\theoremstyle{thmstyletwo}%
\newtheorem{example}{Example}%
\theoremstyle{thmstylethree}%

\begin{document}

\title[]{The $\ell$-intersection Pairs of Constacyclic and Conjucyclic Codes}

%%==================================================================================%%
%% Prefix	-> \pfx{Dr}
%% GivenName	-> \fnm{Joergen W.}
%% Particle	-> \spfx{van der} -> surname prefix
%% FamilyName	-> \sur{Ploeg}
%% Suffix	-> \sfx{IV}
%% NatureName	-> \tanm{Poet Laureate} -> Title after name
%% Degrees	-> \dgr{MSc, PhD}
%% \author*[1,2]{\pfx{Dr} \fnm{Joergen W.} \spfx{van der} \sur{Ploeg} \sfx{IV} \tanm{Poet Laureate} 
%%                 \dgr{MSc, PhD}}\email{iauthor@gmail.com}
%%==================================================================================%%

\author*[1]{\fnm{Md Ajaharul} \sur{Hossain}}\email{mdajaharul@iiitnr.edu.in}

\author[1]{\fnm{Ramakrishna} \sur{Bandi}}\email{ramakrishna@iiitnr.edu.in}
\equalcont{These authors contributed equally to this work.}

%\author[2]{\fnm{Sanjit} \sur{Bhowmick}}\email{sanjitbhowmick392@gmail.com}
%\equalcont{These authors contributed equally to this work.}

\affil*[1]{\orgdiv{Department of Science and Applied Mathematics}, \orgname{IIIT Naya Raipur}, \orgaddress{ \city{Atal Nagar Nava Raipur}, \postcode{493661}, \state{Chhattisgarh}, \country{India}}}

%\affil[2]{\orgdiv{Mathematics}, \orgname{DSPM IIIT-NR}, \orgaddress{\street{Sector 24}, \city{Naya Raipur}, \postcode{493661}, \state{Chhattisgarh}, \country{India}}}

%\affil[2]{\orgdiv{Department of Mathematics}, \orgname{NIT Durgapur}, \orgaddress{ \city{Durgapur}, \postcode{713209}, \state{West Bengal}, \country{India}}}

%%=======================================%%
%% sample for unstructured abstract %%
%%=======================================%%

\abstract{A pair of linear codes whose intersection is of dimension $\ell$, where $\ell$ is a non-negetive integer, is called an $\ell$-intersection pair of codes. This paper focuses on studying $\ell$-intersection pairs of $\lambda_i$-constacyclic, $i=1,2,$ and conjucyclic codes. We first characterize an $\ell$-intersection pair of $\lambda_i$-constacyclic codes. A formula for $\ell$ has been established in terms of the degrees of the generator polynomials of $\lambda_i$-constacyclic codes. This allows obtaining a condition for $\ell$-linear complementary pairs (LPC) of constacyclic codes. Later, we introduce and characterize the $\ell$-intersection pair of conjucyclic codes over $\mathbb{F}_{q^2}$. The first observation in the process is that there are no non-trivial linear conjucyclic codes over finite fields. So focus on the characterization of additive conjucyclic (ACC) codes. We show that the largest $\mathbb{F}_q$-subcode of an ACC code over $\mathbb{F}_{q^2}$ is cyclic and obtain its generating polynomial. This enables us to find the size of an ACC code. Furthermore, we discuss the trace code of an ACC code and show that it is cyclic. Finally, we determine $\ell$-intersection pairs of trace codes of ACC codes over $\mathbb{F}_4$.}


%%\pacs[JEL Classification]{D8, H51}

%%\pacs[MSC Classification]{35A01, 65L10, 65L12, 65L20, 65L70}

\maketitle

\section{Introduction}
Let $q$ be a prime power number and $\mathbb{F}_q$ denote the finite field with $q$ elements. $\mathbb{F}_q^n$, the set of all vectors of length $n$ over $\mathbb{F}_q$, is a vector space of dimension $n$ over $\mathbb{F}_q$. A non-empty subset $C$ of $\mathbb{F}_q^n$ is called a code of length $n$. The code $C$ is called a linear code if it is a subspace of $\mathbb{F}_q^n$. If the dimension of the linear code $C$ is $k$, i.e., $k=dim(C)$, then we say $C$ is an $[n,k]$ linear code over $\mathbb{F}_q$. Besides linear codes, additive codes, the sub-class of linear codes, are a useful family of codes.  An additive subgroup of $\mathbb{F}_{q^2}^n$ is said to be an additive code of length $n$ over $\mathbb{F}_{q^2}$. Unlike linear codes, additive codes do not support scalar multiplications.  An additive code of length $n$ over $\mathbb{F}_{q^2}$ can be seen as an $\mathbb{F}_q$ subspace of $\mathbb{F}_{q^2}^n$. The dimension of an additive code $C$ is defined as the dimension of $C$ as an $\mathbb{F}_q$-subspace of $\mathbb{F}_{q^2}^n$ (which we will refer to as the $\mathbb{F}_q$-dimension). The $\mathbb{F}_q$ basis of $C$ is called a generator matrix of additive code $C$.

Let $u=(u_0,u_1, \ldots, u_{n-1})$ and $v=(v_0,v_1, \ldots, v_{n-1})$ be two elements of $\mathbb{F}_q^n$. Then we define the Hamming distance $d$ between  $u$ and $v$ as the number of coordinates at which $u$ and $v$ differ, i.e., $d=\mid\{ i: u_i \neq v_i ~\forall ~0\leq i\leq n-1 \}\mid$. The minimum Hamming distance of code $C$ is the smallest Hamming distance among all pairs of distinct codewords in $C$. A linear code of length $n$, dimension $k$, and minimum distance $d$ is denoted by $[n,k,d]$, while an additive code of length $n$, $\mathbb{F}_q$-dimension $k$, and minimum distance $d$ is denoted by $(n,k,d)$. Both linear and additive codes have rich algebraic structures and are widely used in the applications of coding theory \cite{calderbank1998quantum, costello1982error, delsarte1998association}. We now define the inner product of two vectors $u$ and $v$ in $\mathbb{F}_q^n$ as $[u,v] _e=\sum_{i=0}^{n-1} u_iv_i$, and is known as the Euclidean inner product. The non-empty subset $C^{\perp_e}= \{x \in \mathbb{F}_q^n~:~ [x,u]_e=0~ \forall ~u \in C \}$ is called the Euclidean dual of the code $C$. In addition to the Euclidean inner product, several other inner products such as Hermitian, Galois, Symplectic, and Alternating inner products are studied in the literature \cite{li2018hermitian, fan2017galois,lv2020algebraic}. 

%The Symplectic and Alternating inner products are used to study conjucyclic codes in this paper. 

Let $C$ be a code and $C^{\perp}$ its dual with respect to any of the above-mentioned inner products. Then $C$ is called self-orthogonal if $C\subseteq C^{\perp}$, and self-dual if $C=C^{\perp}$. A linear code $C$ is known as a linear complementary dual (LCD) code if $C\cap C^\perp={0}$. Massey introduced the concept of LCD codes in \cite{Massey1992}, where he provided a  characterization of LCD codes. Later, Yang and Massey studied cyclic LCD codes and presented a necessary and sufficient condition for a cyclic code to have a complementary dual in \cite{yang1994condition}. Carlet and Guilley discovered that LCD codes are useful in countering certain non-invasive types of security threats, such as Side-Channel Attacks (SCA) and Fault Injection Attacks (FIA) \cite{CG16}. LCD codes are also directly generalized as linear complementary pair (LCP) codes. Subsequently, research on LCD codes has expanded considerably. Two linear codes $C$ and $D$ that satisfy the condition $C \oplus D = \mathbb{F}_q^n$ are referred to as LCP codes,  and $\oplus$ is the direct sum. In \cite{Carlet2018}, Carlet et al. demonstrated the effectiveness of LCP codes in combating security threats such as SCA, FIA, and HTH. They also constructed several good LCP codes. In recent years, LCD and LCP codes have been well understood over finite fields and finite rings \cite{Bhowmick, liu2020sigma, guneri2020linear, hu2021linear}. Recently, in \cite{liu2023lcp}, Liu et al. defined  $\ell$-LCP of codes as an immediate generalization of LCP of codes. The pair $(C_1,C_2)$ is called an $\ell$-LCP of codes if $\dim(C_1\cap C_2)=\ell$ and $C_1+C_2=\mathbb{F}_q^n$. More recently, Guenda et al. \cite{Guenda2019} introduced a more general notion $\ell$-intersection pair of codes over a finite field. The pair $(C_1,C_2)$ is called an $\ell$ intersection pair of codes if $\dim(C_1\cap C_2)=\ell$.

\begin{itemize}
\item If $\ell \neq 0$ with $C_1+C_2=\mathbb{F}_q^n$, then $\ell$ intersection pair is referred as $\ell$-LCP.
\item If $\ell = 0$ with $C_1+C_2=\mathbb{F}_q^n$, then $\ell$ intersection pair is referred as LCP.
\item If $C_2=C_1^{\perp}$, then $\ell$ intersection pair is referred as a hull of $C_1$.
\item If $\ell = 0$ and $C_2=C_1^{\perp}$, then $C_2$ is LCD of $C_1$, and vice versa.
\end{itemize} 

 Conjucyclic codes are a significant class of additive codes with good algebraic structure. An additive code $C$ of length $n$ over $\mathbb{F}_{q^2}$ is called a conjucyclic code if, for any codeword $c=(c_0,c_1,\ldots, c_{n-1})\in C$, the conjucyclic shift $T(c)=(\bar{c}_{n-1}, c_0, \ldots, c_{n-2})$, where $\bar{c}_{i}$ is the conjugate of $c_{i}$ in $\mathbb{F}_{q^2}$, also belongs to $C$. If $a=(a_0,a_1, \ldots, a_{n-1})$ is the generator vector for the additive conjucyclic code $C: (n,q^k)$, then the generator matrix for $C$ is $G={\begin{bmatrix}
a & T(a) & \ldots & T^{k-1}(a)
\end{bmatrix}}^T.$

We now extend the notion of $\ell$-intersection pair of codes to additive conjucyclic codes in particular. Let $D_1$ and $D_2$ be two additive codes over $\mathbb{F}_{q^2}$. Then $(D_1,D_2)$ is an additive $\ell$ intersection pair of codes if and only if $\dim_{\mathbb{F}_q}(D_1\cap D_2)=\ell$. In this paper, we characterize $\ell$ intersection pair of conjucyclic codes.

The paper is organized as follows: Section 2 presents some basic definitions and results required to understand the concepts of other sections. Section 3 characterizes an $\ell$-intersection pair of $\lambda_i$-constacyclic codes (for different values of $i$). Also, a condition for LCP of constacyclic codes is obtained. A necessary and sufficient condition is obtained for an $\ell$ intersection pair of constacyclic codes with the pair of their duals. Then the $\ell$-LCP, the generalization of LCP of codes, is discussed for constacyclic codes. Section 4 characterizes an $\ell$-intersection pair of ACC codes over $\mathbb{F}_{q^2}$. We showed that the largest $\mathbb{F}_q$ subcode of a conjucyclic code over $\mathbb{F}_{q^2}$ is cyclic and also obtained the generating polynomial of the code. Furthermore, we discuss the trace code of an ACC code and show that it is cyclic. Finally, we provide a characterization for $\ell$-intersection pairs of trace codes of ACC codes. Section 5 gives a conclusion for the work discussed in the paper.

\section{Preliminary}

A linear $[n,k]$ code $C$ over $\mathbb{F}_q$ is $\lambda$-constacyclic if, for every $(c_0,c_1,\ldots,c_{n-1})\in C$, $(\lambda c_{n-1},c_0,\ldots,c_{n-2})\in C$, where $0\neq \lambda \in \mathbb{F}_q$. Constacyclic codes generalize cyclic codes, and this is evident when the value of $\lambda$ is set to 1 in the definition of constacyclic codes. Like cyclic codes, every $\lambda$-constacyclic code is isomorphic to an ideal of the quotient ring $R=\mathbb{F}_q[x]/\langle x^n-\lambda \rangle$ as every vector $(a_0,a_1,\ldots, a_{n-1})\in \mathbb{F}_q^n$ can be identified as a polynomial $a_0+a_1x+\ldots+a_{n-1}x^{n-1}$ in $R$. Since $R$ is a principal ideal ring, there exists a unique monic least degree polynomial $g(x)$ in $C$, called the generating polynomial of $C$, such that $C=\langle g(x) \rangle$. If  $g(x)$ is the generating polynomial of a $k$-dimensional constacyclic code $C$, then the degree of $g(x)$ is $n-k$, i.e., $\operatorname{deg}\left(g(x)\right)=n-k$, and $\{g(x),xg(x),\ldots,x^{n-k-1}g(x)\}$ forms a basis for $C$. The generating polynomial $g(  x)$ of a constacyclic code $C$ of length $n$ divides $x^n-\lambda$. Hence, to construct a constacyclic code $C$ of length $n$ and dimension $k$, find a polynomial of degree $n-k$ that divides $x^n-\lambda$. Let $x^n-\lambda=g(x)h(x)$; then, $h(x)$ is called the parity polynomial for $C$. The degree of the parity polynomial $h(x)$ is $\deg (h(x))=k$.

It is well known that the intersection of two cyclic codes $C_1=\langle g_1\rangle$ and $C_2=\langle g_2\rangle$ is also a cyclic code and is generated by $\operatorname{lcm}(g_1,g_2)$ \cite{huffman2010fundamentals}. Similarly, if $C_1=\langle g_1\rangle$ and $C_2=\langle g_2\rangle$ are two $\lambda$-constacyclic codes, then $C_1\cap C_2$ is also an $\lambda$-constacyclic code and is generated by $\operatorname{lcm}(g_1,g_2)$. However, if $C_1$ and $C_2$ are $\lambda_1$- and $\lambda_2$-constacyclic codes, respectively, and $\lambda_1\neq \lambda_2$, then $C_1\cap C_2$ is not necessarily a constacyclic code. This claim can be verified by the following example.
\begin{example}\label{lconsta}
    Consider  $\mathbb{F}_4=\{0,1,\omega,\omega^2=\omega+1\}$, where $\omega^3=1$. Suppose that $C_1=\langle g_1 \rangle$ is a $\omega$-constacyclic code of length $7$, and $C_2=\langle g_2 \rangle$ is a $\omega^2$-constacyclic code of length $7$ over $\mathbb{F}_4$. 
    
    We have 
    $$x^7-\omega=(x+\omega)(x^3 + \omega^2x + 1)(x^3 + \omega x^2 + 1)$$
    $$x^7-\omega^2=(x+\omega^2)(x^3+\omega x+1)(x^3 + \omega^2 x^2 + 1).$$
    
Let $g_1(x)=(x^3 + \omega^2x + 1)$ and $g_2(x)=(x^3 + \omega^2 x^2 + 1)$. Then we can see that  $C_1\cap C_2=\{(0,0,0,0,0,0,0),(1,\omega^2,\omega^2,\omega,\omega^2,\omega^2,1),(\omega,1,1,\omega^2,1,1,\omega),(\omega^2,\omega,\omega,1,\omega.\omega,\omega^2)\},$ a linear code generated by $(1,\omega^2,\omega^2,\omega,\omega^2,\omega^2,1)$, and $C_1\cap C_2$ is neither $\omega$-constacyclic nor $\omega^2$-constacyclic (not even a cyclic code).
\end{example}

 %Two elements $a$ and $b$ of a finite field are called conjugates if they are the roots of the same minimal polynomial over the base field. Let $\mathbb{F}_q$ be the prime field of order $q$. The conjugate of $a\in \mathbb{F}_{q^2}$ is denoted by $\Bar{a}$ and is equal to $a^q$. Conversely, the conjugate of $a^q$ is $a$. 

%We define a Trace mapping $tr: \mathbb{F}_{q^2} \mapsto \mathbb{F}_q$ such that $tr(a_i)=a_i+\Bar{a}_i$. The trace mapping is naturally extended to $\mathbb{F}_{q^2}^n$; that is, $tr(a)=(tr(a_0),tr(a_1), \ldots, tr(a_{n-1}))\in \mathbb{F}_{q}^n$, where $a=(a_0,a_1, \ldots, a_{n-1})\in \mathbb{F}_{q^2}^n$.

%A linear $[n,k,d]$ code is said to be a maximum distance separable (MDS) code if it attains the Singleton bound, i.e., $d \leq n-k+1$. If $C_1$ and $C_2$ are two cyclic codes with generating polynomial $g_1(x)$ and $g_2(x)$, respectively, then $C_1\cap C_2$ is generated by the polynomial $\operatorname{lcm}(g_1(x), g_2(x))$. If $C_1\cap C_2$ is an MDS code, then the minimum distance $d$ of $C_1\cap C_2$ is $\deg(\operatorname{lcm}\left(g_1(x),g_2(x))\right)+1$. In the next section, a characterization of $\ell$-intersection pairs of cyclic codes is given.

%%%%%%%%%%%%%%%%%%%%%%%%%%%%%%%%
\section{A characterization of $\ell$-intersection pairs of Constacyclic codes}

In \cite[Theorem 2.4]{hu2023eaqec}, the authors have studied the $\ell$-intersection pair of constacyclic codes and obtained the value of $\ell$. However, they have not explored much. This section provides a complete characterization of the $\ell$-intersection pair of constacyclic codes, such as examining the existence $\ell$-intersection of constacyclic codes, the relation between $\ell$-intersection pair of codes and the intersection of their duals, etc. 

In \cite[Theorem 1]{hossain2023linear}, we have provided a necessary and sufficient condition for the existence of $\ell$-intersection pair of cyclic codes. The same results hold true for $\lambda$-constacyclic codes. We present the same as the following theorem.
\begin{theorem}
\label{thcharr}
    Let $C_i=\langle g_i \rangle$ be a $[n, k_i]$ $\lambda$-constacyclic code with corresponding parity polynomial $h_i$ over $\mathbb{F}_q$, for $i=1,2$. Then  $(C_1, C_2)$ is an $\ell$-intersection pair of codes if and only if  
\[ \operatorname{deg} \operatorname{ \operatorname{lcm} } (g_2,h_1) = n-k_2+ \ell,\]
and 
\[ \operatorname{deg} \operatorname{lcm} (g_1, h_2) = n-k_1 + \ell. \]
\end{theorem}
\begin{proof}
The proof is similar to that of \cite[Theorem 1]{hossain2023linear}. 
\end{proof}
We, now, focus on the characterization of $\ell$ intersection pair of $\lambda_i$-constacyclic codes (constacyclic codes with two different $\lambda$ values). In \cite[Theorem 2.4]{hu2023eaqec}, the authors obtained the value of  $\ell$. We present the value of $\ell$ in terms of the degrees of generator and parity polynomials of contacyclic codes (Theorem \ref{differentconstac}). 
\begin{theorem}\cite[Theorem 2.4]{hu2023eaqec}
\label{sumfull}
  Let $C_i$ be a $\lambda_i$-constacyclic code of length $n$, over $\mathbb{F}_q$, for $i=1,2$, such that $\lambda_1\neq \lambda_2$ and $C_1\cap C_2\neq \{0\}$. Then $\dim(C_1\cap C_2)=\dim(C_1)+\dim(C_2)-n$.  
\end{theorem}
% \begin{lemma}\cite[Corollary 2.3]{hu2023eaqec}
%     \label{lemcon}
% Let $C_i$ be a $\lambda_i$-constacyclic code of length $n$, over $\mathbb{F}_q$, for $i=1,2$, such that $\lambda_1\neq \lambda_2$. If $k_1+k_2>n$, then $C_1\cap C_2\neq \{0\}$.
% \end{lemma}
\begin{lemma}
\label{lemspe}\cite[Proposition 2.3]{liu2022note}
  Let $C_i$ be a $\lambda_i$-constacyclic code of length $n$, over $\mathbb{F}_q$, for $i=1,2$, such that $\lambda_1\neq \lambda_2$. If $C_1\cap C_2\neq \{0\}$, then $C_1+C_2=\mathbb{F}_q^n$.    
\end{lemma}
% \begin{theorem}
% \label{differentconstac}
% Let $C_i=\langle g_i \rangle$ be a $\lambda_i$-constacyclic code with parameter $[n,k_i]$, over $\mathbb{F}_q$, for $i=1,2$, and $\ell=\dim(C_1\cap C_2)\neq 0$. Also, let $h_i=(x^n-\lambda_i)/g_i$ be a parity polynomial corresponding to $g_i$. If 
% $\lambda_1 \neq \lambda_2$, then $\ell=\deg h_2- \deg g_1$.
%  \end{theorem}
%  \begin{proof}
%      We have $\dim(C_1)=n-\deg g_1$ and $\dim(C_2)=n-\deg g_2=\deg h_2$. Hence the theorem.
%  \end{proof}
\begin{theorem}
\label{differentconstac}
Let $C_i=\langle g_i \rangle$ be a $\lambda_i$-constacyclic code with parameter $[n,k_i]$, over $\mathbb{F}_q$, for $i=1,2$, and $\ell=\dim(C_1\cap C_2)$. Also, let $h_i=(x^n-\lambda_i)/g_i$ be a parity polynomial corresponding to $g_i$. If 
$\lambda_1 \neq \lambda_2$, then
 \[
    \ell = \begin{cases}
        0, & \text{for } k_1+k_2\leq n\\
       \deg h_2- \deg g_1, & \text{for } k_1+k_2> n.
        \end{cases}
 \]
\end{theorem}
\begin{proof}
Suppose that $\lambda_1 \neq \lambda_2$. Consider $k_1+k_2 \leq n$, i.e, $k_1+k_2-n\leq 0$. If possible, let $\ell> 0$. Then from Lemma \ref{lemspe}, we have $C_1+C_2=\mathbb{F}_q^n$, i.e. $\dim(C_1+C_2)=n$. For the linear codes $C_1$ and $C_2$, we have $\dim(C_1+C_2)=\dim(C_1)+\dim(C_2)-\dim(C_1\cap C_2)$. Therefore $k_1+k_2-n=\ell$, and as $\ell>0$ by assumption, therefore $k_1+k_2-n>0$, which is a contradiction to $k_1+k_2 \leq n$. Thus $\ell=0$.

% Let $k_1+k_2>n$. We have $\dim(C_1)=n-\deg g_1$ and $\dim(C_2)=n-\deg g_2=\deg h_2$. From Theorem \ref{sumfull} and Lemma \ref{lemcon}, we have $\ell=\deg h_2- \deg g_1$ for $k_1+k_2>n$.

Let $k_1+k_2>n$, i.e., $k_1+k_2-n>0$. Since $C_1,C_2$ and $C_1+C_2$ are subspaces of $\mathbb{F}_q^n$ so $\dim(C_1+C_2)\leq n$. The equation $\dim(C_1+C_2)=\dim(C_1)+\dim(C_2)-\dim(C_1\cap C_2)$ gives $\ell=k_1+k_2-\dim(C_1+C_2)\geq k_1+k_2-n>0$, i.e, $\ell>0$. From Lemma \ref{lemspe}, we have $C_1+C_2=\mathbb{F}_q^n$, i.e., $\dim(C_1+C_2)=n$. Since $C_1=\langle g_1 \rangle$ and $C_2=\langle g_2 \rangle$ therefore $k_1=n-\deg g_1=\deg h_1$ and  $k_2=n-\deg g_2=\deg h_2$. Now, As $\dim(C_1+C_2)=\dim(C_1)+\dim(C_2)-\dim(C_1\cap C_2)$, then $\ell=k_1+k_2-n=k_2-(n-k_1)=\deg h_2-\deg g_1$.
\end{proof}

\begin{example}
\label{charexam}
In continuation of \ref{lconsta}, 
%Let $C_1$ and $C_2$ be a $\omega$ and $\omega^2$ constacyclic code of length $7$, respectively, over the field $\mathbb{F}_4=\{0,1,\omega,\omega^2\}$, $1+\omega+\omega^2=0$. We have
%\[x^7-\omega=(x+\omega)(x^3+\omega^2x+1)(x^3+\omega x^2+1),\]
%\[x^7-\omega^2=(x+\omega^2)(x^3+\omega x+1)(x^3+\omega^2 x^2+1).\]
the intersection of $C_1=\langle g_1 \rangle$ and $C_2=\langle g_2\rangle$, where $g_1=(x^3+\omega x^2+1)$ and $g_2=(x^3+\omega x+1)$ is $C_1\cap C_2=\operatorname{Span}_{\mathbb{F}_4}\{(1,\omega,\omega,\omega^2,\omega,\omega,1)\}$. This implies that $\ell=1$. 

Since $\deg h_2=\deg (x^n-\omega^2)/g_2=4$, and $\deg g_1=3$ so $\deg h_2-\deg g_1=4-3=1=\ell$.
%\end{example}
%\begin{example}
%\label{charexam}

Similarly, consider 
%$C_1$ and $C_2$ be a $\omega$ and $\omega^2$ constacyclic code of length $7$, respectively, over the field $\mathbb{F}_4=\{0,1,\omega,\omega^2\}$, $1+\omega+\omega^2=0$. We have
%\[x^7-\omega=(x+\omega)(x^3+\omega^2x+1)(x^3+\omega x^2+1),\]
%\[x^7-\omega^2=(x+\omega^2)(x^3+\omega x+1)(x^3+\omega^2 x^2+1).\]
 $C_3=\langle g_3 \rangle$ and $C_4=\langle g_4\rangle$, where $g_3=(x+\omega)(x^3+\omega x^2+1)$ and $g_4=(x+\omega^2)(x^3+\omega x+1)$.  Then $C_3\cap C_4=\{0\}$. This implies that $\ell=0$. Here $k_3=3$ and $k_4=3$, so $k_3+k_4=6<n=7$.
\end{example}

We now prove that  $\lambda_i$-constacyclic codes is an LCP of codes using Theorem \ref{differentconstac}.

\begin{corollary}
  Let $C_i$ be a $\lambda_i$-constacyclic code with parameter $[n,k_i]$, over $\mathbb{F}_q$, for $i=1,2$, such that $\lambda_1\neq \lambda_2$ and $k_1+k_2=n$. Then $(C_1, C_2)$ is an LCP of codes.  
\end{corollary}
\begin{proof}
Suppose that $C_i$ are $\lambda_i$-constacyclic code with parameter $[n,k_i]$, over $\mathbb{F}_q$, for $i=1,2$, such that $\lambda_1\neq \lambda_2$ and $k_1+k_2=n$. Then from Theorem \ref{differentconstac}, we have 
$\ell=\dim(C_1\cap C_2)=0$. So $C_1+C_2=\mathbb{F}_q^n$. Thus $(C_1, C_2)$ is an LCP of codes.
\end{proof}

\begin{example}
Consider a $\omega$-constacyclic code $C_1$ and a $\omega^2$-constacyclic code $C_2$ of length $15$ over the field $\mathbb{F}_4=\{0,1,\omega,\omega^2\}$, $1+\omega+\omega^2=0$. We have
\[x^{15}-\omega=(x^3 + \omega^2)(x^6 + x^3 +\omega)(x^6 + \omega x^3 + \omega),\]
\[x^{15}-\omega^2=(x^3 + \omega)(x^6 + x^3 +\omega^2)(x^6 + \omega^2 x^3 + \omega^2).\]

Let $C_1=\langle g_1 \rangle$ and $C_2=\langle g_2\rangle$, where $g_1=(x^3 + \omega^2)(x^6 + x^3 +\omega)$ and $g_2=(x^6 + \omega^2 x^3 + \omega^2)$. Then $k_1=6$ and $k_2=9$, which implies $k_1+k_2=15$. We can see that $C_1\cap C_2=\{0\}$, and so $\ell=0$. Hence $(C_1, C_2)$ is an LCP of codes.
\end{example}
 We now prove that  a $\lambda$-constacyclic code is LCD code using Theorem \ref{differentconstac}.
\begin{corollary}
    Let C be a $\lambda$-constacyclic code over $\mathbb{F}_q$ such that $\lambda^2\neq 1$. Then $C$ is an LCD code.
\end{corollary}
\begin{proof}
 Let C be a $\lambda$-constacyclic code over $\mathbb{F}_q$ such that $\lambda^2\neq 1$. From \cite[Proposition 2.4]{dinh2010constacyclic}, $C^\perp$ is a $\lambda^{-1}$-constacyclic code. Since $\lambda^2\neq 1$, i.e, $\lambda\neq \lambda^{-1}$ and $\dim(C)+\dim(C^\perp)=n$, therefore from the Theorem \ref{differentconstac}, $\ell=\dim(C\cap C^\perp)=0$, i.e, $C$ is an LCD code.
\end{proof}

In the rest of the section, we examine the intersection of the duals of an  $\ell$-intersection pair of $\lambda_i$-constacyclic codes. This helps to establish a condition for $\ell$-LCP of codes.

\begin{lemma}\cite[Theorem 3]{hossain2023linear}
\label{dualin}
Let $C_i$ be a linear code of length $n$ and dimension $k_i$, respectively, over $\mathbb{F}_q$, for $i=1,2$. If $(C_1,C_2)$ be a $\ell$-intersection pair of codes, then $(C_1^\perp,C_2^\perp)$ is an $\left(n-(k_1+k_2-\ell)\right)$-intersection pair of codes.
\end{lemma}
\begin{theorem}
\label{ldiffconsta}
Let $C_i$ be a $\lambda_i$-constacyclic code with parameter $[n,k_i]$ over $\mathbb{F}_q$, for $i=1,2$, such that $\lambda_1 \neq \lambda_2$. If $(C_1,C_2)$ be a non-trivial $\ell$-intersection pair of codes, then $(C_1^\perp,C_2^\perp)$ is a trivial $\ell$-intersection pair of codes.
\end{theorem}

\begin{proof}
Let us assume that $(C_1,C_2)$ is a non-trivial $\ell$-intersection pair of codes, i.e. $\ell\neq 0$. From Lemma \ref{lemspe}, we have $C_1+C_2=\mathbb{F}_q^n$, this implies that $\dim(C_1+C_2)=n$. Thus we have $\dim(C_1)+\dim(C_2)-\dim(C_1\cap C_2)=n$, i.e, $k_1+k_2-\ell=n$. From Lemma \ref{dualin}, $(C_1^\perp,C_2^\perp)$ is a trivial $\ell$-intersection pair of codes.
\end{proof}
In general, the converse of the above theorem is not true, i.e., if $(C_1^\perp, C_2^\perp)$ is a trivial $\ell$-intersection pair of codes then it is not necessary that $(C_1, C_2)$ is a non-trivial $\ell$-intersection pair of codes. This can be seen in the following example.
\begin{example}
 Consider $C_1$ and $C_2$ are $\omega$ and $\omega^2$ constacyclic code of length $7$, respectively, over the $\mathbb{F}_4=\{0,1,\omega,\omega^2\}$, $1+\omega+\omega^2=0$. From \cite[Proposition 2.4]{dinh2010constacyclic}, $C_1^\perp$ and $C_2^\perp$ are $\omega^2$ and $\omega$ constacyclic codes, respectively, of length $7$. We have
\[x^7-\omega=(x+\omega)(x^3+\omega^2x+1)(x^3+\omega x^2+1),\]
\[x^7-\omega^2=(x+\omega^2)(x^3+\omega x+1)(x^3+\omega^2 x^2+1).\]
Let $C_1=\langle g_1 \rangle$ and $C_2=\langle g_2\rangle$, where $g_1=(x+w)(x^3+\omega x^2+1)$ and $g_2=(x^3+\omega x+1)$. Then $h_1=\frac{x^7-\omega}{g_1}=(x^3 + \omega^2 x +1)$ and $h_2=\frac{x^7-\omega^2}{g_2}=(x + \omega^2)(x^3 + \omega^2 x^2 +1)$ are the parity polynomials for $C_1$ and $C_2$, respectively. The monic reciprocal polynomial of $h_1$ and $h_2$ are $h_1^\ast=(x^3 + \omega^2 x^2 +1)$ and $h_2^\ast=(x + \omega)(x^3 + \omega^2 x +1)$, respectively. Therefore $C_1^\perp=\langle (x^3 + \omega^2 x^2 +1) \rangle$ and $C_2^\perp=\langle(x + \omega)(x^3 + \omega^2 x +1)\rangle$. We can see that from Theorem \ref{differentconstac}, $C_1\cap C_2=\{0\}$ and $C_1^\perp\cap C_2^\perp=\{0\}$ as $k_1+k_2 = n$.
\end{example}
However, the converse of Theorem \ref{ldiffconsta} is true under a condition. This allowed us to obtain a necessary and sufficient condition to determine the relationship between an $\ell$-intersection pair of $\lambda_i$-constacyclic codes and the intersection of their duals.
\begin{theorem}
    \label{ldiffiff}
Let $C_i$ be a $\lambda_i$-constacyclic code with parameter $[n,k_i]$ over $\mathbb{F}_q$, for $i=1,2$, such that $\lambda_1 \neq \lambda_2$ and $k_1+k_2\neq n$. Then $(C_1,C_2)$ is a non-trivial $\ell$-intersection pair of codes if and only if $(C_1^\perp,C_2^\perp)$ is a trivial $\ell$-intersection pair of codes. Moreover, either $(C_1, C_2) $ or $(C_1^{\perp}, C_2^{\perp})$ is trivial (non-trivial) intersection.
\end{theorem}
\begin{proof} Consider that $C_i$ is a $\lambda_i$-constacyclic code  over $\mathbb{F}_q$ such that $\lambda_1 \neq \lambda_2$ and $k_1+k_2\neq n$. Since the necessary part follows from Theorem \ref{ldiffconsta}, we prove the sufficient part.

    %Suppose that $(C_1,C_2)$ is a non-trivial $\ell$-intersection pair of codes. Then from Theorem \ref{ldiffconsta}, $(C_1^\perp,C_2^\perp)$ is a trivial $\ell$-intersection pair of codes.

    Assume that $(C_1^\perp,C_2^\perp)$ is a $\ell^\prime$ intersection pair of codes  with $\ell^\prime=0$ (trivial intersection). Let $k_i^\prime=n-k_i$ be the dimensions of $C_i^{\perp}$. Then from Theorem \ref{differentconstac}, $k_1^\prime+k_2^\prime \leq n$. This implies that $k_1+k_2 >n$. Again from Theorem \ref{differentconstac}, $\ell \neq 0$. Therefore $(C_1,C_2)$ is a non-trivial $\ell$-intersection pair of codes.

   % Applying Lemma \ref{dualin} on $C_i^\perp$, we get that $(C_1,C_2)$ is a $\left(n-(k_1^\prime+k_2^\prime-\ell^\prime)\right)$ intersection pair of codes, i.e., $(k_1+k_2-n)$ intersection pair of codes. Since $k_1+k_2\neq n$ therefore $(C_1,C_2)$ is a non-trivial $\ell$-intersection pair of codes.
\end{proof}

\begin{theorem}
    Let $C_i$ be a $\lambda_i$-constacyclic code with parameter $[n,k_i]$ over $\mathbb{F}_q$, for $i=1,2$, such that $\lambda_1 \neq \lambda_2$ and $k_1+k_2\neq n$. Then either $(C_1, C_2)$ or $(C_1^{\perp}, C_2^{\perp})$ is an $\ell$-LCP of codes. 
\end{theorem}

\begin{proof}
 Let $C_i$ be a $\lambda_i$-constacyclic code with parameter $[n,k_i]$ over $\mathbb{F}_q$, for $i=1,2$, such that $\lambda_1 \neq \lambda_2$ and $k_1+k_2\neq n$. From Theorem \ref{ldiffiff}, either $(C_1^\perp, C_2^\perp)$ or $(C_1,C_2)$ is non-trivial intersection. Further from Lemma \ref{lemspe}, $C_1+C_2=\mathbb{F}_q^n$ or $C_1^{\perp}+C_2^{\perp}=\mathbb{F}_q^n$, i.e., $(C_1,C_2)$ or $(C_1^\perp, C_2^\perp)$ is an $\ell$-LCP of codes.
\end{proof}
\begin{example}
Consider $C_1$ and $C_2$ be a $\omega$ and $\omega^2$ constacyclic code of length $7$, respectively, over the field $\mathbb{F}_4=\{0,1,\omega,\omega^2\}$, $1+\omega+\omega^2=0$. From \cite[Proposition 2.4]{dinh2010constacyclic}, $C_1^\perp$ and $C_2^\perp$ are $\omega^2$ and $\omega$ constacyclic codes, respectively, of length $7$. We have
\[x^7-\omega=(x+\omega)(x^3+\omega^2x+1)(x^3+\omega x^2+1),\]
\[x^7-\omega^2=(x+\omega^2)(x^3+\omega x+1)(x^3+\omega^2 x^2+1).\]
Let $C_1=\langle g_1 \rangle$ and $C_2=\langle g_2\rangle$, where $g_1=(x^3+\omega x^2+1)$ and $g_2=(x^3+\omega x+1)$. Then $k_1=4$ and $k_2=4$, so $k_1+k_2\neq n$. From Theorem \ref{differentconstac}, $\ell\neq 0$. Therefore from Lemma \ref{lemspe}, $C_1+C_2=\mathbb{F}_q^n$, i.e., $(C_1,C_2)$ is an $\ell$-LCP of codes.

% Then $h_1=\frac{x^7-\omega}{g_1}=(x + \omega)(x^3 + \omega^2 x +1)$ and $h_2=\frac{x^7-\omega^2}{g_2}=(x + \omega^2)(x^3 + \omega^2 x^2 +1)$ are the parity polynomials for $C_1$ and $C_2$, respectively. The monic reciprocal polynomial of $h_1$ and $h_2$ are $h_1^\ast=(x + \omega^2)(x^3 + \omega^2 x^2 +1)$ and $h_2^\ast=(x + \omega)(x^3 + \omega^2 x +1)$, respectively. Therefore $C_1^\perp=\langle (x + \omega^2)(x^3 + \omega^2 x^2 +1) \rangle$ and $C_2^\perp=\langle(x + \omega)(x^3 + \omega^2 x +1)\rangle$. We can see that $C_1\cap C_2=\operatorname{Span}_{\mathbb{F}_4}\{(1,\omega,\omega,\omega^2,\omega,\omega,1)\}$, i.e., $(C_1,C_2)$ is a non-trivial $\ell$-intersection pair, and $C_1^\perp\cap C_2^\perp=\{0\}$. Hence $(C_1^\perp,C_2^\perp)$ is a trivial $\ell$-intersection pair of codes.
\end{example}
%%%%%%%%%%%%%%%%%%%%%%%%%%%%%%%%%%%%%%%%%%% 

%%%%%%%%%%%%%%%%%

%%%%%%%%%%%%%%%%

%%%%%%%%%%%%
%%%%%%%%%%%%%%%%%%%% Conclusion  %%%%%%%%%%%%%%%%%%%%%%%%%%%%%%%%%%%%%%%%%%%%%%

\section{$\ell$-intersection pairs of additive conjucyclic codes }
In this section, we explore the $\ell$-intersection of additive conjucyclic codes. Recall that a additive conjucyclic code of length $n$ over $\mathbb{F}_q^2$ is an additive code in which for any $c=(c_0,c_1, \ldots, c_{n-1}) \in C$, its concyclic shift $T(c)=(\Bar{c}_{n-1}, c_0, \ldots, c_{n-2}) \in C$,  where  $\Bar{c}_{i}$ is the conjugate of $c_i$ over $\mathbb{F}_q$. If $a=(a_0,a_1, \ldots, a_{n-1})$ is the generator vector for the additive conjucyclic code $C: (n,q^k)$, then the generator matrix for $C$ is $G={\begin{bmatrix}
a & T(a) & \ldots & T^{k-1}(a)
\end{bmatrix}}^T$. 

To define the dual of an additive conjucyclic code over $\mathbb{F}_q^2$, we use three inner products: Euclidean, Symplectic, and Alternating inner products. These types of inner products are already discussed in \cite{lv2020algebraic}. 
%To prove this section's main theorem, Symplectic and Alternating inner products are required. The Symplectic inner product is defined over the vectors of even length. 

Consider two vectors $u=(u_0,u_1,\cdots,u_{n-1})$ and $v=(v_0,v_1,\cdots,v_{n-1})$ in $\mathbb{F}_{q^2}^n$. Euclidean, Symplectic, and Alternating inner products are denoted by ${\langle u,v \rangle}_e$, ${\langle u,v \rangle}_s$, and ${\langle u,v \rangle}_a$, respectively, and defined by
\begin{eqnarray*}
{\langle u,v \rangle}_e &=& \sum_{i=0}^{n-1} u_iv_i, \\
{\langle u,v \rangle}_s &=& \sum_{i=0}^{m-1} (u_iv_{m+i}-u_{m+i}v_i), ~\text{where} ~ n=2m,\\
{\langle u,v \rangle}_a &=& (\Bar{\alpha}^2-\alpha^2) \sum_{i=0}^{n-1} (u_i\Bar{v}_i-\Bar{u}_iv_i), ~\text{respectively},
\end{eqnarray*}
where $\alpha$ is a primitive element in the finite field $\mathbb{F}_{q^2}$.

The dual of an additive code $C$ with respect to the inner products, Euclidean, Symplectic, and alternating, are defined as $C^{\perp_x}=\{u \in \mathbb{F}_{q^2}^n : \langle u, ~v \rangle_x = 0 \}$, where $x=e, a,s$. Note that the Euclidean and Symplectic duals of an additive code $C$ over $\mathbb{F}_{q^2}$ are $\mathbb{F}_{q^2}$-linear. 
\begin{proposition}
\label{conjutrivial}
    The only linear conjucyclic codes over $\mathbb{F}_{q^2}$ are $\{0\}$ and $\mathbb{F}_{q^2}^n$.
\end{proposition}
\begin{proof}
Suppose that $C$ is a non-zero linear conjucyclic code of length $n$ over $\mathbb{F}_{q^2}$. For any $0\neq a=u+\alpha v=(u_0+\alpha v_0,u_1+\alpha v_1,\ldots, u_{n-1}+\alpha v_{n-1})\in C$, where $u,v\in \mathbb{F}_{q}^n$, we have $a-T^n(a) \in C$, where $T$ is the conjucyclic shift. Then $a-T^n(a)=(\alpha-\Bar{\alpha})(v_0,v_1,\ldots,v_{n-1})\in C$. Since $\alpha \neq \Bar{\alpha}$, so $v=(v_0,v_1,\ldots,v_{n-1})\in C$. This implies that $u=a-\alpha v \in C$. 

Since $\Bar{u}_i = u_i$ and $\Bar{v}_i = v_i$ in $\mathbb{F}_q$, $T(u)=\sigma(u)$ and $T(v)=\sigma(v)$, where $\sigma$ is the cyclic shift. Therefore $\sigma(a)=\sigma(u+\alpha v)=\sigma(u)+\alpha \sigma(v) =T(u)+\alpha T(v)=T(u+\alpha v)=T(a) \in C$. Thus $C$ is a cyclic code. 

Since $C$ is a non-zero code, there exist a codeword $c=(c_0,c_1,\ldots,c_{n-1})\in C$ such that some  $c_i \neq 0$ in $\mathbb{F}_{q^2}\setminus \mathbb{F}_{q}$. Thus $\sigma(\sigma^{n-i-1}(c))-T(\sigma^{n-i-1}(c))=(c_i-\Bar{c}_i,0,\ldots,0)=(c_i-\Bar{c}_i)(1,0,\ldots,0) \in C$. Therefore $(1,0,\ldots,0)\in C$. Hence $C=\mathbb{F}_{q^2}^n$.

% We have $\sigma(a)-T(a)=((\alpha-\Bar{\alpha})v_{n-1},0,\ldots,0)=(\alpha-\Bar{\alpha})v_{n-1}(1,0,\ldots,0)\in C$. Since $\alpha \neq \Bar{\alpha}$ and $v_{n-1}\neq 0$ therefore $(1,0,\ldots,0)\in C$. Hence $C=\mathbb{F}_{q^2}^n$.
\end{proof}
\begin{proposition}
The Euclidean and Symplectic duals of an additive code $C$ over $\mathbb{F}_{q^2}$ are $\mathbb{F}_{q^2}$-linear, and thus trivial.
\end{proposition}

\begin{proof}
Let us assume that $u_1,u_2\in C^{\perp_e}$. Then $\langle u_1, ~v \rangle_e =\sum_{i=0}^{n-1} u_{1i}v_i= 0 $ and $\langle u_2, ~v \rangle_e =\sum_{i=0}^{n-1} u_{2i}v_i= 0 $ for all $v\in C$. For every $r_1,r_2\in \mathbb{F}_{q^2}$, we have $\langle r_1u_1+r_2u_2, v \rangle_e=\sum_{i=0}^{n-1} (r_1u_{1i}+r_2u_{2i})v_i=r_1 \sum_{i=0}^{n-1} u_{1i}v_i+r_2 \sum_{i=0}^{n-1} u_{2i}v_i=0$  for all $v\in C$. Thus $r_1u_1+r_2u_2\in C^{\perp_e}$ and so $C^{\perp_e}$ is a linear code. 

Similarly, for Symplectic dual, let $u_1,u_2\in C^{\perp_s}$. Then $\langle u_1, ~v \rangle_s =\sum_{i=0}^{m-1} (u_{1i}v_{m+i}-u_{1,m+i}v_i)=0 $ and $\langle u_2, ~v \rangle_s =\sum_{i=0}^{m-1} (u_{2i}v_{m+i}-u_{2,m+i}v_i)=0$ for all $v\in C$. For every $r_1,r_2\in \mathbb{F}_{q^2}$, we have $\langle r_1u_1+r_2u_2, v \rangle_s=\sum_{i=0}^{m-1} \{(r_1u_{1i}+r_2u_{2i})v_{m+i}-(r_1u_{1,m+i}+r_2 u_{2,m+i} )v_i\}=r_1 \sum_{i=0}^{m-1} (u_{1i}v_{m+i}-u_{1,m+i}v_i)+r_2 \sum_{i=0}^{m-1} (u_{2i}v_{m+i}-u_{2,m+i}v_i)=0$ for all $v\in C$. Thus $C^{\perp_s}$ is a linear code. Hence, from Proposition \ref{conjutrivial}, $C^{\perp_e}$ and $C^{\perp_e}$ are trivial.
\end{proof}

 So we use the Alternating inner product to study the dual of an additive conjucyclic code over $\mathbb{F}_{q^2}$. Also, since there are no non-trivial linear conjucyclic codes over $\mathbb{F}_{q^2}$ so conjucyclic codes simply mean additive conjucyclic codes unless otherwise specified.
 %The Symplectic inner products are used on the linear, cyclic codes of even length over  $\mathbb{F}_{q}$. 

In \cite{lv2020algebraic}, the authors have defined a mapping $\Psi_\alpha$ from $\mathbb{F}_{q^2}^n$ to $\mathbb{F}_{q}^{2n}$ such that for a $u=(u_0,u_1,\cdots,u_{n-1})\in \mathbb{F}_{q^2}^n$, $$\Psi_\alpha(u) = \left( tr(\alpha u_0), \ldots, tr(\alpha u_{n-1}), tr(\Bar{\alpha} u_0), \ldots, tr(\Bar{\alpha} u_{n-1}) \right),$$ where $\alpha$ is a primitive element in $\mathbb{F}_{q^2}$ to give a characterization to conjucyclic codes. We now use the mapping $\Psi_\alpha$ to discuss the $\ell$-intersection pairs of conjucyclic codes. We can see that the map $\Psi_\alpha$ is an $\mathbb{F}_{q}$-linear isomorphism. This provides a relation between conjucyclic code over $\mathbb{F}_{q^2}$ and linear cyclic codes over $\mathbb{F}_{q}$, also a relation between Alternating and Symplectic inner products as stated below.

\begin{theorem}
\cite[Theorem 3.6]{lv2020algebraic}
    \label{linearmap}
    $C$ is an $(n,q^k)$ conjucyclic code over $\mathbb{F}_{q^2}$ if and only if there exist an $[2n,k]$  cyclic code $D$ over $\mathbb{F}_{q}$ such that $\Psi_\alpha(C)=D$, where $\Psi_\alpha$ is the $\mathbb{F}_{q}$-linear isomorphism between $\mathbb{F}_{q^2}^n$ and $\mathbb{F}_{q}^{2n}$.
\end{theorem}

\begin{proposition}
    \cite[Proposition 3.8]{lv2020algebraic}
 \label{innermap}   
 Let $u$ and $v$ be two vectors in $\mathbb{F}_{q^2}^n$. Then
 \[{\langle u,v \rangle}_a={\langle \Psi_\alpha(u),\Psi_\alpha(v) \rangle}_s.\]
\end{proposition}

We define two types of matrix multiplications based on Alternating and Symplectic inner products. Let $A={\begin{bmatrix}
        a_1 & a_2 & \cdots & a_r
    \end{bmatrix}}^T$ and $B={\begin{bmatrix}
        b_1 & b_2 & \cdots & b_t
    \end{bmatrix}}$, where each $a_i$, $b_j$ are vectors of length $n$. Then the matrix multiplications $\odot_a$ and $\odot_s$ are defined by $A \odot_a B=\left( {\langle a_i, b_j \rangle}_a\right)_{r \times t}$ and $A \odot_s B=\left( {\langle a_i, b_j \rangle}_s\right)_{r \times t}$, respectively. Note here that the multiplications  $\odot_a$ and $\odot_s$ are similar to the usual matrix multiplication except that the Alternating and Symplectic inner products replace the Euclidean inner product. So the above matrix multiplications exist if and only if the usual matrix multiplication exists. The image of the matrix $A$ under the map $\Psi_\alpha$ is defined as $\Psi_\alpha (A)={\begin{bmatrix}
        \Psi_\alpha(a_1) & \Psi_\alpha(a_2) & \cdots & \Psi_\alpha(a_r)
    \end{bmatrix}}^T$, where $\Psi_\alpha(a_i)$ is a vector of length $2n$.

  An isomorphism is established between the matrix multiplications $\odot_a$ and $\odot_s$ in the following lemma.

\begin{lemma}
\label{innermatrixp}
    Let $A$ and $B$ be two matrices over $\mathbb{F}_{q^2}$ such that $A \odot_a B$ exists. Then \[A \odot_a B= \Psi_\alpha(A)\odot_s \Psi_\alpha(B).\]
\end{lemma}
\begin{proof}
    Let us assume that $A={\begin{bmatrix}
        a_1 & a_2 & \cdots & a_r
    \end{bmatrix}}^T$ and $B={\begin{bmatrix}
        b_1 & b_2 & \cdots & b_t
    \end{bmatrix}}$. Then $A \odot_a B=\left( {\langle a_i, b_j \rangle}_a\right)$. By Lemma \ref{innermap}, ${\langle a_i, b_j \rangle}_a={\langle \Psi_\alpha(a_i),\Psi_\alpha(b_i) \rangle}_s$. Therefore $A \odot_a B=\left( {\langle \Psi_\alpha(a_i),\Psi_\alpha(b_i) \rangle}_s\right)_{r \times t}=\Psi_\alpha(A)\odot_s \Psi_\alpha(B)$.
\end{proof}

Let $G$ and $H$ be a generator and parity check matrices of a conjucyclic code $C$, respectively. Then $G \odot_a H^T=\textbf{0}$. Throughout the paper, the rank of a matrix $A$ over $\mathbb{F}_{q^2}$ is defined as $\operatorname{rank} (A)=\dim_{\mathbb{F}_{q}}\{a_1,a_2,\cdots,a_r\}$, where $a_i$ is the row-vectors of matrix $A$.

Unlike cyclic codes, the study of conjucyclic codes is primarily based on the generator matrix. The generating polynomials play an important role in cyclic codes, as every cyclic code is isomorphic to an ideal of a quotient ring, whereas conjucyclic codes are yet to be studied through the ring theory approach. So we focus on the generator matrix of conjucyclic codes to characterize an $\ell$-intersection pair of conjucyclic codes.  Recall the definition of additive $\ell$-intersection pair of codes, if two conjucyclic codes $C_1$ and $C_2$ form an $\ell$-intersection pair, over $\mathbb{F}_{q^2}$, we have $\ell=\dim_{\mathbb{F}_q}(C_1\cap C_2)$. In the following theorem we present a formula for the value of $\ell$.
\begin{theorem}
\label{charell}
Suppose $C_i$ be an $(n,q^{k_i})$ conjucyclic code over $\mathbb{F}_{q^2}$ with generator and parity check matrices $G_i$ and $H_i$, respectively, for $i=1,2$. If $(C_1,C_2)$ is an $\ell$-intersection pair, then \[\ell= k_1 - \operatorname{rank} (G_1 \odot_{a} H_2^T)=k_2 - \operatorname{rank}(G_2 \odot_{a} H_1^T).\]
\end{theorem}
\begin{proof}
Consider that  $(C_1,C_2)$ is an $\ell$-intersection pair of conjucyclic codes $C_1$ and $C_2$. So $\dim_{\mathbb{F}_q}(C_1\cap C_2)=\ell$. From Theorem \ref{linearmap}, $ \dim_{\mathbb{F}_q}(C_1+ C_2)\leq 2n$. Thus $k_1+k_2-\ell\leq 2n$, that gives $2n-k_2\geq k_1-\ell$ or $2n-k_1\geq k_2-\ell$. 
%Let $B=\{r_1,r_2,\cdots,r_\ell\}$ be a basis of $C_1\cap C_2$, then $r_i\in C_1$ and $C_2$, for $1\leq i\leq \ell$.  
We know that $C_1\cap C_2 \subseteq C_1$ and $C_1\cap C_2 \subseteq C_2$. 

When $C_1\cap C_2 = C_1 \subseteq C_2$. We have $G_1\odot_{a}H_2^T=\textbf{0}$ and $\ell=k_1$. Therefore  $rank(G_1\odot_{a}H_2^T)=0=k_1-\ell$. Similarly,   when $C_1\cap C_2 = C_2 \subseteq C_1$ we get $rank(G_2\odot_{a}H_1^T)=0=k_2-\ell$. 

We will prove the result  for $C_1\cap C_2 \subset C_i$, i.e., $\ell < k_i$.

From Theorem \ref{linearmap}, there exists an $[2n,k_i]$ linear $q$-ary cyclic code $D_i$ such that $\Psi_\alpha(C_i)=D_i$, where $\Psi_\alpha$ is an $\mathbb{F}_{q}$-linear isomorphism between $C_i$ and $D_i$, for $i=1,2$. Let $G_i^\prime$ and $H_i^\prime$ be  generator and parity check matrices of $D_i$, respectively.

Consider a $\mathbb{F}_q$ basis $B=\{r_1,r_2,\cdots,r_\ell\}$ of $C_1\cap C_2$. Now $B$ can be extended to a $\mathbb{F}_q$ basis $\{r_1,r_2,\cdots,r_\ell,r_{\ell+1},\cdots,r_{k_1}\}$ of $C_1$. Then $M_1=[\begin{matrix} r_1 ~r_2 ~ \dots ~r_\ell ~r_{\ell+1}~ \dots ~ r_{k_1}\end{matrix}]^T$ is a generator matrix for $C_1$, and the generator matrix $G_1=A \odot_a M_1$, where $A$ is an invertible matrix. Let $M_1^\prime$ be the generator matrix, corresponding to $M$, of the cyclic code $D_1$. Then from Lemma \ref{innermatrixp}, $G_1\odot_{a}H_2^T =G_1^\prime \odot_s {H^\prime}_2^T=(A^\prime \odot_s M_1^\prime) \odot_s {H^\prime}_2^T$, where $A^\prime$ is invertible matrix over $\mathbb{F}_{q}$. Therefore $rank(G_1^\prime \odot_s {H^\prime}_2^T)=rank(M_1^\prime \odot_s {H^\prime}_2^T)$. Now
\[G_1\odot_{a}H_2^T=\begin{bmatrix} 0\\ \hline \begin{bmatrix}r_{\ell+1}\\ \vdots \\ r_{k_1}\end{bmatrix}\odot_{a}H_2^T\end{bmatrix}=\begin{bmatrix} 0\\ \hline \begin{bmatrix}{r^\prime}_{\ell+1}\\ \vdots \\ {r^\prime}_{k_1}\end{bmatrix} \odot_s {H^\prime}_2^T\end{bmatrix}\]
The matrix $\begin{bmatrix}{r^\prime}_{\ell+1}\\ \vdots \\ {r^\prime}_{k_1}\end{bmatrix}\odot_s {H^\prime}_2^T$ has order $(k_1-\ell)\times (2n-k_2)$ with $2n-k_2\geq k_1-\ell$, so it follows that $rank\left(\begin{bmatrix}{r^\prime}_{\ell+1}\\ \vdots \\ {r^\prime}_{k_1}\end{bmatrix}\odot_s{H^\prime}_2^T\right)\leq k_1-\ell$. Suppose that $rank\left(\begin{bmatrix}{r^\prime}_{\ell+1}\\ \vdots \\ {r^\prime}_{k_1}\end{bmatrix}\odot_s{H^\prime}_2^T\right)< k_1-\ell$. Then there exists a non-zero vector $v\in \mathbb{F}_q^{k_1-\ell}$ such that $v\odot_s \left(\begin{bmatrix}{r^\prime}_{\ell+1}\\ \vdots \\ {r^\prime}_{k_1}\end{bmatrix}\odot_s{H^\prime}_2^T\right)=\textbf{0}$. Therefore $v \odot_s \begin{bmatrix}{r^\prime}_{\ell+1}\\ \vdots \\ {r^\prime}_{k_1}\end{bmatrix}\in D_2\setminus \{0\}$, a contradiction as $\operatorname{Span} \{{r^\prime}_{\ell+1},{ r^\prime}_{\ell+2},\cdots , {r^\prime}_{k_1}\}\cap D_2=\{0\}$, and so $v\odot_s\begin{bmatrix}{r^\prime}_{\ell+1}\\ \vdots \\ {r^\prime}_{k_1}\end{bmatrix}\notin D_2$. Therefore $\operatorname{rank} (G_1 \odot_{a} H_2^T)=k_1-\ell=\operatorname{rank}(H_2 \odot_{a} G_1^T)$. Hence 
 $\ell= k_1 - \operatorname{rank} (G_1 \odot_{a} H_2^T)=k_2 - \operatorname{rank}(G_2 \odot_{a} H_1^T)$.
\end{proof}

In \cite{liu2023lcp}, Liu et al. defined  $\ell$-LCP of codes as an immediate generalization of LCP of codes. We now define $\ell$-additive complementary pair (ACP) of additive codes $C_1$ and $C_2$  over $\mathbb{F}_{q^2}$ in the similar lines of $\ell$-LCP of codes. Two additive codes $C_1$ and $C_2$ of length $n$ over $\mathbb{F}_{q^2}$ is said to be an $\ell$-ACP if $\dim_{\mathbb{F}_q}(C_1\cap C_2)=\ell$ and $C_1+C_2=\mathbb{F}_{q^2}^n$. Using Theorem \ref{charell}, we obtain a condition on  $\ell$-ACP of conjucyclic codes over $\mathbb{F}_{q^2}$.
\begin{theorem}
 \label{llcp}  
 For $i\in \{1,2\}$, let $C_i$ be an $(n,q^{k_i})$ conjucyclic code over $\mathbb{F}_{q^2}$ with generator and parity check matrices $G_i$ and $H_i$, respectively. If $(C_1,C_2)$ is $\ell$-ACP, then $\ell=k_1+k_2 -2n$ and $\operatorname{rank} (G_1 \odot_{a} H_2^T)=2n-k_1$.
\end{theorem}
\begin{proof}
    Let us assume that $(C_1,C_2)$ is an $\ell$-ACP of conjucyclic codes $C_i$. Then $C_1+C_2=\mathbb{F}_{q^2}^n$. Therefore we have $\dim_{\mathbb{F}_q}(C_1)+\dim_{\mathbb{F}_q}(C_2)-\dim_{\mathbb{F}_q}(C_1\cap C_2)=\dim_{\mathbb{F}_q}(C_1+C_2)=2n$ i.e., $k_1+k_2-\ell=2n$. Rest follows from Theorem \ref{charell}.
\end{proof}
We provide an example to illustrate our discussion.
\begin{example}
\label{egell}
    Let $C_1$ be the $(7,2^7)$ conjucyclic code over $\mathbb{F}_4$ with the following generator matrix $$G_1=\begin{bmatrix}
        1 & 1+\omega & 1+\omega & 1+\omega &0 &0& 1+\omega\\
        \omega & 1&1+\omega &1+\omega&1+\omega&0&0 \\
        0 & \omega &1 &1+\omega&1+\omega&1+\omega &0 \\
        0&0&\omega&1&1+\omega&1+\omega&1+\omega\\
        \omega &0&0&\omega&1&1+\omega&1+\omega \\
        \omega &\omega &0&0&\omega &1&1+\omega \\
        \omega&\omega&\omega&0&0&\omega&1
    \end{bmatrix}.$$
    Again let $C_2$ be the $(7,2^5)$ conjucyclic code over $\mathbb{F}_4$ with the following parity check matrix
    $$H_2=\begin{bmatrix}
        1+\omega & 1+\omega & 1+\omega& 0&0& 1+\omega &0\\
        0 &1+\omega&1+\omega&1+\omega&0&0&1+\omega \\
        \omega&0&1+\omega&1+\omega&1+\omega&0&0\\
        0&\omega&0&1+\omega&1+\omega&1+\omega&0\\
        0&0&\omega&0&1+\omega&1+\omega&1+\omega\\
        \omega&0&0&\omega&0&1+\omega&1+\omega\\
        \omega &\omega&0&0&\omega&0&1+\omega\\
       \omega & \omega &\omega&0&0&\omega&0\\
       0&\omega & \omega &\omega&0&0&\omega
    \end{bmatrix}.$$
  Therefore $G_1\odot_a H_2^T=G_1 \Bar{H}_2^T+ \Bar{G}_1H_2^T=\begin{bmatrix}
      \omega&0&1+\omega&1+\omega&1+\omega& 0&0&1+\omega&0\\
      0&\omega&0&1+\omega&1+\omega& 1+\omega&0&0&1+\omega\\
      0&0&\omega&0&1+\omega&1+\omega&1+\omega&0&0\\
      \omega&0&0&\omega&0&1+\omega&1+\omega&1+\omega&0\\
      \omega&\omega&0&0&\omega&0&1+\omega&1+\omega&1+\omega\\
      \omega&\omega&\omega&0&0&\omega&0&1+\omega&1+\omega\\
      0&\omega&\omega&\omega&0&0&\omega&0&1+\omega
  \end{bmatrix}+\begin{bmatrix}
      1 + \omega & 0 & \omega & \omega & \omega&0&0&\omega&0\\
      0 & 1 + \omega & 0 & \omega & \omega&\omega&0&0&\omega\\
      0 & 0 & 1 +\omega & 0 & \omega&\omega&\omega&0&0\\
      1 + \omega & 0 & 0 & 1 + \omega & 0&\omega&\omega&\omega&0\\
      1 + \omega & 1 + \omega & 0 & 0 & 1 + \omega&0&\omega&\omega&\omega\\
      1 + \omega & 1 + \omega & 1 + \omega & 0 & 0&1+\omega&0&\omega&\omega\\
      0 & 1 + \omega & 1 + \omega & 1 + \omega & 0&0&1+\omega&0&\omega
  \end{bmatrix}=\begin{bmatrix}
      1&0&1&1&1&0&0&1&0\\
      0&1&0&1&1&1&0&0&1\\
      0&0&1&0&1&1&1&0&0\\
      1&0&0&1&0&1&1&1&0\\
      1&1&0&0&1&0&1&1&1\\
      1&1&1&0&0&1&0&1&1\\
      0&1&1&1&0&0&1&0&1\\
  \end{bmatrix}.$  
  Further obtaining row reduced echelon form we have $\operatorname{rank} (G_1 \odot_{a} H_2^T)=3$. By Magma algebra computational software $\ell=\dim_{\mathbb{F}_2}(C_1\cap C_2)=4$. Here $k_1=7$ and therefore  $\operatorname{rank} (G_1 \odot_{a} H_2^T)=k_1-\ell$.
\end{example}
\section{Trace of a conjucyclic code}

In this section, we discuss the trace of a conjucyclic code $C$ over $\mathbb{F}_{q^2}$. This enables us to estimate the size of a conjucyclic code. Recall the definition, the trace of an element $a\in \mathbb{F}_{q^2}$ is $tr(a)=a+\Bar{a}\in \mathbb{F}_q$, where $\Bar{a}=a^q$ is the conjugate element of $a$ upon the field $ \mathbb{F}_q$. The trace code of a code $C$ is defined as $$tr(C)=\{(tr(c_0),tr(c_1),\ldots,tr(c_{n-1})) ~\mid~ (c_0,c_1,\ldots,c_{n-1})\in C \}.$$

\begin{theorem}
\label{lacccyclic}
Let $C$ be a conjucyclic code of length $n$ over $\mathbb{F}_{q^2}$. Then $tr(C)$ is a cyclic code of length $n$ over $\mathbb{F}_q$.
\end{theorem}
\begin{proof}
 For any $c=(c_0,c_1,\cdots,c_{n-1})$ in conjucyclic code $C$, we have $T(c) \in C$ and $tr(c)\in tr(C)$. Let $\sigma$ be the cyclic shift operation. Then
  $\sigma(tr(c))=\sigma(tr(c_0,c_1,\cdots,c_{n-1}))
 =\sigma(c_0+\bar{c}_0,c_1+\bar{c}_1,\cdots,c_{n-1}+\bar{c}_{n-1})
 =(c_{n-1}+\bar{c}_{n-1},c_0+\bar{c}_0,\cdots,c_{n-2}+\bar{c}_{n-2})
 =tr(\bar{c}_{n-1},c_0,\cdots,c_{n-2})=tr(T(c)) \in tr(C)$.
\end{proof}

The following lemma gives the relation between the $\ell$ intersection of trace codes and the trace of $\ell$ intersection of codes.
\begin{lemma}
\label{treq}
    Let $C$ and $D$ be two conjucyclic codes over the finite field $\mathbb{F}_{q^2}$. Then $tr(C\cap D)=tr(C)\cap tr(D)$.
\end{lemma}
\begin{proof}
    Consider that $c\in tr(C\cap D)$, this implies that $c=tr(d)$ for some $d\in C\cap D$. Therefore $d\in C$ and also $d\in D$, this gives $tr(d)\in tr(C)$ and $tr(d)\in tr(D)$. Thus $c=tr(d)\in tr(C)\cap tr(D)$ and this results $tr(C\cap D)\subseteq tr(C)\cap tr(D)$.

    Conversely, let $e\in tr(C)\cap tr(D)$, then $e\in tr(C)$ and $e\in tr(D)$. If possible let $tr(s)=e\notin tr(C\cap D)$. 
    %Therefore $tr(s)\notin tr(C\cap D)$ for any $s\in \mathbb{F}_{q^2}$ satisfying $e=tr(s)$. 
    Then $s\notin C\cap D$. So $s\notin C$ or $s\notin D$. If $s\notin C$ then $e=tr(s)\notin tr(C)$, a contradiction. Again if $s\notin D$ then $e=tr(s)\notin tr(D)$, a contradiction. Thus $e\in tr(C\cap D)$. Threfore $tr(C)\cap tr(D)\subseteq tr(C\cap D)$.
\end{proof}
%Each element $a\in \mathbb{F}_{q^2}^n$ can be written as $a=(u_0+\alpha v_0, u_1+\alpha v_1,\ldots, u_{n-1}+\alpha v_{n-1})$, where $u_i,v_i\in \mathbb{F}_{q}$ for $i=0,1,\ldots,n-1$ and $\alpha$ is a primitive element in $\mathbb{F}_{q^2}$.
In the following theorem we prove that every conjucyclic code contains its trace code.

\begin{theorem}
\label{trrwc}
    Let $C$ be a conjucyclic code of length $n$ over $\mathbb{F}_{q^2}$, then $tr(C)\subseteq C.$
\end{theorem}
\begin{proof}
  Let $tr(c)=(tr(c_0),tr(c_1),\ldots,tr(c_{n-1}))\in tr(C)$, where $c=(c_0,c_1,\ldots,c_{n-1})\in C$. Since $C$ is a conjucyclic code so $T^n(c)=(\Bar{c}_0,\Bar{c}_1,\ldots,\Bar{c}_{n-1})\in C$. Therefore $c+T^n(c)=(c_0+\Bar{c}_0,c_1+\Bar{c}_1,\ldots,c_{n-1}+\Bar{c}_{n-1})\in C$, i.e. $(tr(c_0),tr(c_1),\ldots,tr(c_{n-1}))=tr(c)\in C$. Thus $tr(C)\subseteq C.$
\end{proof}
In the following theorem, a condition for trivial $\ell$-intersection pair of trace codes of conjucyclic codes is obtained.  
\begin{corollary}
    Let $(C, D)$ be a trivial $\ell$-intersection pair of conjucyclic codes. Then $(tr(C), tr(D))$ is a trivial $\ell$-intersection pair of cyclic codes.
\end{corollary}
\begin{proof}
    Follows from Lemma \ref{treq} and Theorem \ref{trrwc}.
    %we have $tr(C)\cap tr(D)=tr(C\cap D)\subseteq C\cap D=\{0\}$. Hence the theorem.
\end{proof}
We denote $\mathcal{S}_C$, the largest $\mathbb{F}_q$-subcode of a conjucyclic code $C$ over $\mathbb{F}_{q^2}.$ This enables us to find the $\mathbb{F}_q$-dimension of a trace code of a conjucyclic code, then the size of a conjucyclic code.  In the following lemma, we show that the largest $\mathbb{F}_q$-sub code $\mathcal{S}_C$ contained in a conjucyclic code $C$ is a cyclic code.
\begin{lemma}
\label{lsubcode}
    Let $C$ be a conjucyclic code over $\mathbb{F}_{q^2}.$ Then the $\mathbb{F}_q$-subcode  $\mathcal{S}_C$ is a cyclic code.
\end{lemma}
\begin{proof}
    Let $c=(c_0,c_1,\ldots,c_{n-1})\in \mathcal{S}_C\subseteq C$, where $c_i\in \mathbb{F}_q$. Then the conjucyclic shit $T(c)=(\Bar{c}_{n-1},c_0,\ldots,c_{n-2})=(c_{n-1},c_0,\ldots,c_{n-2})=\sigma(c)\in \mathcal{S}_C$ as $\Bar{a}=a^q=a$ for all $a\in \mathbb{F}_q$. Hence $\mathcal{S}_C$ is a cyclic code.
\end{proof}
Define the mapping $\Psi_{\alpha,p}:\mathbb{F}_{q^2}[x]/(x^n-1)\rightarrow \mathbb{F}_{q}[x]/(x^{2n}-1)$ such that $\Psi_{\alpha,p}(u_0+u_1x+\ldots+u_{n-1}x^{n-1})=tr(\alpha u_0)+tr(\alpha u_1)x+\ldots+tr(\alpha u_{n-1})x^{n-1}+tr(\Bar{\alpha} u_0)x^n+\ldots+ tr(\Bar{\alpha} u_{n-1})x^{2n-1}$, where $\alpha$ is a primitive element in the finite field $\mathbb{F}_{q^2}$. Clearly $\Psi_{\alpha,p}$ is an $\mathbb{F}_q$ linear isomorphism as $\Psi$ is an $\mathbb{F}_q$ linear isomorphism. Let us define $S_D=\Psi(S_C)$. Since $tr(a)=tr(\Bar{a})$ for all $a\in \mathbb{F}_{q^2}$, then 
Define the mapping $\Psi_{\alpha,p}:\mathbb{F}_{q^2}[x]/(x^n-1)\rightarrow \mathbb{F}_{q}[x]/(x^{2n}-1)$ such that $\Psi_{\alpha,p}(u_0+u_1x+\ldots+u_{n-1}x^{n-1})=tr(\alpha u_0)+tr(\alpha u_1)x+\ldots+tr(\alpha u_{n-1})x^{n-1}+tr(\Bar{\alpha} u_0)x^n+\ldots+ tr(\Bar{\alpha} u_{n-1})x^{2n-1}$, where $\alpha$ is a primitive element in the finite field $\mathbb{F}_{q^2}$. Clearly $\Psi_{\alpha,p}$ is an $\mathbb{F}_q$ linear isomorphism as $\Psi$ is an $\mathbb{F}_q$ linear isomorphism. Let us define $S_D=\Psi(S_C)$. Since $tr(a)=tr(\Bar{a})$ for all $a\in \mathbb{F}_{q^2}$, then 
Define the mapping $\Psi_{\alpha,p}:\mathbb{F}_{q^2}[x]/(x^n-1)\rightarrow \mathbb{F}_{q}[x]/(x^{2n}-1)$ such that $\Psi_{\alpha,p}(u_0+u_1x+\ldots+u_{n-1}x^{n-1})=tr(\alpha u_0)+tr(\alpha u_1)x+\ldots+tr(\alpha u_{n-1})x^{n-1}+tr(\Bar{\alpha} u_0)x^n+\ldots+ tr(\Bar{\alpha} u_{n-1})x^{2n-1}$, where $\alpha$ is a primitive element in the finite field $\mathbb{F}_{q^2}$. Clearly $\Psi_{\alpha,p}$ is an $\mathbb{F}_q$ linear isomorphism as $\Psi$ is an $\mathbb{F}_q$ linear isomorphism. Let us define $S_D=\Psi(S_C)$. Since $tr(a)=tr(\Bar{a})$ for all $a\in \mathbb{F}_{q^2}$, then 
Define the mapping $\Psi_{\alpha,p}:\mathbb{F}_{q^2}[x]/(x^n-1)\rightarrow \mathbb{F}_{q}[x]/(x^{2n}-1)$ such that $\Psi_{\alpha,p}(u_0+u_1x+\ldots+u_{n-1}x^{n-1})=tr(\alpha u_0)+tr(\alpha u_1)x+\ldots+tr(\alpha u_{n-1})x^{n-1}+tr(\Bar{\alpha} u_0)x^n+\ldots+ tr(\Bar{\alpha} u_{n-1})x^{2n-1}$, where $\alpha$ is a primitive element in the finite field $\mathbb{F}_{q^2}$. Clearly $\Psi_{\alpha,p}$ is an $\mathbb{F}_q$ linear isomorphism as $\Psi_\alpha$ is an $\mathbb{F}_q$ linear isomorphism. Let us denote $S_D=\Psi_\alpha(S_C)$. Then 
    \[S_D=\{tr(\alpha)(u \mid u): ~\mbox{for all}~u\in S_C\},\] where $\mid$ is concatenation of vectors, 
    as $tr(a)=tr(\Bar{a})$ for all $a\in \mathbb{F}_{q^2}$.
\begin{theorem}
  Let $C$ be a non-trivial conjucyclic code over $\mathbb{F}_{q^2}$ with its largest $\mathbb{F}_q$-sub-code $S_C$. Then $S_D=\Psi(S_C)$ is a cyclic code.
\end{theorem}
\begin{proof}
    Let $\mathbf{u}=tr(\alpha)(u\mid u)\in S_D$. Then $\sigma(\mathbf{u})=tr(\alpha)(\sigma(u)\mid \sigma(u))$. From Lemma \ref{lsubcode}, $S_C$ is a cyclic code, therefore $\sigma(u)\in S_C$, and thus $\sigma(\mathbf{u})\in S_D$. Hence the theorem.
\end{proof}
Let $u(x)=u_0+u_1x+\ldots+u_{n-1}x^{n-1}$ be the generating polynomial of $S_C$, then $u(x)\mid x^n-1$. Since $\Psi_\alpha,p$ is an $\mathbb{F}_q$ linear isomorphism so $\operatorname{gcd}(\Psi_{\alpha,p}(u(x)),x^{2n}-1)$ is the generating polynomial for $S_D$. We have $\Psi_{\alpha,p}(u(x))=tr(\alpha)(u(x)+x^n u(x))=tr(\alpha)(x^n+1)u(x)$. 
Since $u(x)\mid x^n-1$, so $(x^n+1)u(x)\mid x^{2n}-1$. Therefore $(x^n+1)u(x)$ is the generating polynomial for $S_D$. From Theorem \ref{linearmap}, $D=\Psi_\alpha(C)=\langle \mathbf{g}(x) \rangle$, and $S_D$ is a subcode of the cyclic code $D$. Therefore $\mathbf{g}(x)\mid (x^n+1)u(x)$.

A vector $g=(g_0,g_1,\ldots,g_{n-1})$ is said to be a generating vector of the non-trivial conjucyclic code $C$ over $\mathbb{F}_{q^2}$ if and only if $\langle \mathbf{g}(x) \rangle=D=\Psi_\alpha(C)=\langle \Psi_{\alpha,p}(g(x))\rangle$, where $g(x)=g_0+g_1x+\ldots+g_{n-1}x^{n-1}$. Since $S_D \subset D$, and $S_D$ and $D$ are cyclic, so $\Psi_{\alpha,p}(g(x))$ divides $(x^n+1)u(x)$. Therefore we have 
   \[\frac{\Psi_{\alpha,p}(g(x))}{\operatorname{gcd}(\Psi_{\alpha,p}(g(x)),x^n+1)}\mid u(x).\] Again, since $S_C$ is the largest $\mathbb{F}_q$-subcode of $C$, $ u(x)=\frac{\Psi_{\alpha,p}(g(x))}{\operatorname{gcd}(\Psi_{\alpha,p}(g(x)),x^n+1)}$. Hence the following theorem, in which we present the generating polynomial of the largest $\mathbb{F}_q$ subcode $S_C$ of a non-trivial conjucyclic code $C$. 

\begin{theorem}
\label{largestgen}
   Let $C$ be a non-trivial conjucyclic code over $\mathbb{F}_{q^2}$ of length $n$ with generating vector $g$. Then the largest $\mathbb{F}_q$-subcode $S_C$ of $C$ is generated by the polynomial $\frac{\Psi_{\alpha,p}(g(x))}{\operatorname{gcd}(\Psi_{\alpha,p}(g(x)),x^n+1)}$.  
\end{theorem}
% \begin{proof}
% Let $S_C=\langle u(x)\rangle$, then from the above discussion $S_D=\Psi(S_C)=\langle (x^n+1)u(x)\rangle$ and $u(x)\mid x^n-1$. Since $S_D$ is a subcode $D$ so the generating polynomial of $D$, i.e., $\Psi_{\alpha,p}(g(x))$ divides the polynomial $(x^n+1)u(x)$. Therefore we have 
%    \[\frac{\Psi_{\alpha,p}(g(x))}{\operatorname{gcd}(\Psi_{\alpha,p}(g(x)),x^n+1)}\mid u(x).\]
% We claim that $u(x)=\frac{\Psi_{\alpha,p}(g(x))}{\operatorname{gcd}(\Psi_{\alpha,p}(g(x)),x^n+1)}.$ If possible let $u(x)\neq \frac{\Psi_{\alpha,p}(g(x))}{\operatorname{gcd}(\Psi_{\alpha,p}(g(x)),x^n+1)}$. Since $u(x)\mid x^n-1$ and $\frac{\Psi_{\alpha,p}(g(x))}{\operatorname{gcd}(\Psi_{\alpha,p}(g(x)),x^n+1)}\mid u(x)$ so $S_C$ is a subcode of the $\mathbb{F}_q$ cyclic code generated by $\frac{\Psi_{\alpha,p}(g(x))}{\operatorname{gcd}(\Psi_{\alpha,p}(g(x)),x^n+1)}$, which is a contradiction that $S_C$ is the largest $\mathbb{F}_q$-subcode of $C$.
% \end{proof}

In the following theorem, we determine the $\mathbb{F}_2$-dimension of the trace code $tr(C)$ of a conjucyclic code $C$ over $\mathbb{F}_4$. 
\begin{theorem}
\label{dimtr}
 Let $C$ be a conjucyclic code over $\mathbb{F}_{4}$ and  $\mathcal{S}_C$ be its largest $\mathbb{F}_2$-subcode. Then $$\operatorname{dim}_{\mathbb{F}_2}(tr(C))=\operatorname{dim}_{\mathbb{F}_2}(C)-\operatorname{dim}_{\mathbb{F}_2}(\mathcal{S}_C).$$
\end{theorem}
\begin{proof}
Define a mapping $tr: C \mapsto \mathbb{F}_2^n$ such that $tr(c)=(tr(c_0),tr(c_1),\ldots,tr(c_{n-1}))\in \mathbb{F}_2^n$, where $c=(c_0,c_1,\ldots,c_{n-1})\in C$. Clearly, the `$tr$' mapping is a linear mapping, and the kernel of the mapping is  $\operatorname{ker}(tr)=\{u\in C ~\mid ~ tr(u)=\mathbf{0}\}$. 

Since $S_C$ is the largest $\mathbb{F}_q$-sub code in $C$, $S_C \subseteq \operatorname{ker}(tr)$. Similarly,  for any $v=(v_0,v_1,\ldots,v_{n-1})\in C\setminus \mathcal{S}_C$, we have $tr(v)\neq \mathbf{0}$ as there exist at least one $v_i\in \mathbb{F}_{4}\setminus \mathbb{F}_2$ in a co-ordinate position of $v$ such that $v_i+\Bar{v}_i\neq 0$. Thus  $\operatorname{ker}(tr) \subseteq \mathcal{S}_C$. Therefore $\operatorname{ker}(tr) = \mathcal{S}_C$. By the first isomorphism theorem of linear algebra, we have $C/\operatorname{ker}(tr)$ isomorphic to $tr(C)$. Thus $\operatorname{dim}_{\mathbb{F}_2}(C)-\operatorname{dim}_{\mathbb{F}_2}(\mathcal{S}_C)=\operatorname{dim}_{\mathbb{F}_2}(tr(C))$.
\end{proof}
\begin{remark}
    If $C$ is a conjucyclic code over $\mathbb{F}_{4}$ and  $\mathcal{S}_C$ its largest $\mathbb{F}_2$-subcode. Then $\mid C \mid=\mid tr(C)\mid \cdot \mid \mathcal{S}_C \mid$.
\end{remark}

\begin{lemma}
\label{lsubin}
    Let $C_i$ be a conjucyclic code over $\mathbb{F}_{q^2}$ and  $\mathcal{S}_{C_i}$ be its largest $\mathbb{F}_q$-subcode. Then $\mathcal{S}_{C_1\cap C_2}=\mathcal{S}_{C_1}\cap \mathcal{S}_{C_2}$.
\end{lemma}
\begin{proof}
    Let us assume that $c=(c_0,c_1,\ldots,c_{n-1})\in \mathcal{S}_{C_1\cap C_2}$. Since $\mathcal{S}_{C_1\cap C_2}$ is the largest $ \mathbb{F}_q$sub code in $C$, so $c_i\in \mathbb{F}_q$ and $c\in C_1\cap C_2$. This implies that $c\in C_1$ and $c\in C_2$. So we have $c\in \mathcal{S}_{C_1}$ and $c\in \mathcal{S}_{C_2}$. Therefore $c\in \mathcal{S}_{C_1}\cap \mathcal{S}_{C_2}$. Hence $\mathcal{S}_{C_1\cap C_2}\subseteq\mathcal{S}_{C_1}\cap \mathcal{S}_{C_2}$. The other inclusion follows in the same way.

   % Conversely, assume that $c=(c_0,c_1,\ldots,c_{n-1})\in \mathcal{S}_{C_1}\cap \mathcal{S}_{C_2}$. Then $c_i\in \mathbb{F}_q$, for $i=0,1,\ldots,n-1$ and $c\in \mathcal{S}_{C_1}$ and $c\in \mathcal{S}_{C_2}$. This implies that $c\in C_1\cap C_2$ with $c_i\in \mathbb{F}_q$, for $i=0,1,\ldots,n-1$. So we have $c\in \mathcal{S}_{C_1\cap C_2}$. Thus $\mathcal{S}_{C_1}\cap \mathcal{S}_{C_2}\subseteq\mathcal{S}_{C_1\cap C_2}$.
\end{proof}
In the following theorem, using the fact that the largest $\mathbb{F}_2$-subcode $\mathcal{S}_C$ is cyclic from Lemma \ref{lsubcode}, a characterization is obtained for the $\ell$ intersection pair of trace codes of two conjucyclic codes.
\begin{theorem}
    Let $C_i$ be an $(n,2^{k_i})$ conjucyclic code over $\mathbb{F}_{4}$ with generator and parity check matrix $G_i$ and $H_i$, respectively, and $\mathcal{S}_{C_i}=\langle g_i(x)\rangle$ be its largest $\mathbb{F}_2$-subcode, for $i=1,2$. Then $(tr(C_1), tr(C_2))$ is an $\ell$ intersection pair of cyclic codes, where $$\ell=\operatorname{deg}\operatorname{lcm}(g_1,g_2)-\operatorname{rank} (G_1 \odot_{a} H_2^T)-(n-k_1),$$
    also $$\ell=\operatorname{deg}\operatorname{lcm}(g_1,g_2)-\operatorname{rank} (G_2 \odot_{a} H_1^T)-(n-k_2).$$
\end{theorem}
\begin{proof}
 %From Theorem \ref{dimtr}, we have $\operatorname{dim}_{\mathbb{F}_2}(tr(C_1\cap C_2))=\operatorname{dim}_{\mathbb{F}_2}(C_1\cap C_2)-\operatorname{dim}_{\mathbb{F}_2}(\mathcal{S}_{C_1\cap C_2})$. From Lemma \ref{lsubcode} and Lemma \ref{lsubin}, $\mathcal{S}_{C_1\cap C_2}$ is a cyclic code with generating polynomial $\operatorname{lcm}(g_1,g_2)$, i.e., $\mathcal{S}_{C_1\cap C_2}=\langle \operatorname{lcm}(g_1,g_2) \rangle$. Therefore $\operatorname{dim}_{\mathbb{F}_2}(\mathcal{S}_{C_1\cap C_2})=n-\operatorname{deg} \operatorname{lcm}(g_1,g_2)$. From Lemma \ref{treq}, $\operatorname{dim}_{\mathbb{F}_2}(tr(C_1\cap C_2))=\operatorname{dim}_{\mathbb{F}_2}(tr(C_1)\cap tr(C_2))=\ell$. From Theorem \ref{charell}, $\operatorname{dim}_{\mathbb{F}_2}(C_1\cap C_2)=k_1-\operatorname{rank} (G_1 \odot_{a} H_2^T)=k_2-\operatorname{rank} (G_2 \odot_{a} H_1^T)$. Thus $\ell=\operatorname{deg}\operatorname{lcm}(g_1,g_2)-\operatorname{rank} (G_1 \odot_{a} H_2^T)-(n-k_1)$, also $\ell=\operatorname{deg}\operatorname{lcm}(g_1,g_2)-\operatorname{rank} (G_2 \odot_{a} H_1^T)-(n-k_2).$

We can see that
\begin{eqnarray*}
    \ell & = & \operatorname{dim}_{\mathbb{F}_2}(tr(C_1\cap C_2)) \\
    &=& \operatorname{dim}_{\mathbb{F}_2}(tr(C_1)\cap tr(C_2)) \\
    &=& \operatorname{dim}_{\mathbb{F}_2}(C_1\cap C_2)-\operatorname{dim}_{\mathbb{F}_2}(\mathcal{S}_{C_1\cap C_2}) - \operatorname{dim}_{\mathbb{F}_2}(\mathcal{S}_{C_1\cap C_2}) ~
    (\mbox{From Theorem \ref{dimtr}})\\
    &=& k_1-\operatorname{rank} (G_1 \odot_{a} H_2^T) - \operatorname{dim}_{\mathbb{F}_2}(\mathcal{S}_{C_1} \cap \mathcal{S}_{C_2})
    ~(\mbox{From Theorem \ref{charell} and Lemma \ref{lsubcode}})\\
    &=& k_1-\operatorname{rank} (G_1 \odot_{a} H_2^T) - (n-\operatorname{deg} \operatorname{lcm}(g_1,g_2)) \\
    &=& \operatorname{deg}\operatorname{lcm}(g_1,g_2)-\operatorname{rank} (G_1 \odot_{a} H_2^T)-(n-k_1).
\end{eqnarray*}
\end{proof}
\begin{example}
 Consider two conjucyclic codes over $\mathbb{F}_{4}$ from the Example \ref{egell}. Here $n=7,k_1=7$ and $k_2=5$. From Example \ref{egell}, $\operatorname{rank} (G_1 \odot_{a} H_2^T)=3$. From Theorem \ref{largestgen}, we obtain $g_1(x)=x^3+x+1$ and $g_2(x)=x^3+x+1$. By Magma computational software, we have $\ell=0$, i.e., $(tr(C_1), tr(C_2))$ is a trivial intersection pair of codes. We can see $\operatorname{deg}\operatorname{lcm}(g_1,g_2)-\operatorname{rank} (G_1 \odot_{a} H_2^T)-(n-k_1)=3-3-(7-7)=0$.
\end{example}
\section{Conclusion}
In this paper, we have studied $\ell$-intersection pair of constacyclic and conjucyclic codes over finite fields. First, we gave a characterization of $\ell$ intersection of $\lambda_i$-constacyclic codes for $i=1,2$. From the characterization, a condition for LCP of codes and thereafter a condition for $\ell$-LCP, generalization of LCP of codes are obtained. The LCP of codes is useful in combating some cryptographic attacks \cite{Carlet2018}, whereas $\ell$ intersection pair of some cyclic codes are useful in the construction of some good EAQEC codes\cite{hossain2023linear}.  It will be interesting to construct EAQECC codes through our study.

In addition to that, we introduced the definition of $\ell$-intersection pair of additive codes over $\mathbb{F}_{q^2}$ and gave a characterization for $\ell$-intersection pair of additive 
conjucyclic codes over $\mathbb{F}_{q^2}$. Furthermore, we have studied the trace code and largest $\mathbb{F}_{q}$ subcode of a conjucyclic code and showed that they are cyclic.  Determined the size of an ACC code using the help of generator polynomial of the largest $\mathbb{F}_q$-subcode of an ACC code, and the traceo of the ACC code. Finally, gave a characterization for $\ell$-intersection pairs of trace codes of conjucyclic codes over $\mathbb{F}_4$.

 \section*{Statements and Declarations}
\subsection*{Ethical Approval}
We would confirm that this paper has not been published nor submitted for publication elsewhere. We confirm that we have read, understand, and agreed to the submission guidelines, policies, and submission declaration of the journal.

\subsection*{Competing Interests}
The authors have no conflict of interest to declare.

\subsection*{Authors' contributions}
Conceptualization: [Md Ajaharul Hossain, Ramakrishna Bandi], Methodology: [Md Ajaharul Hossain, Ramakrishna Bandi], Formal analysis and investigation: [Md Ajaharul Hossain, Ramakrishna Bandi], Writing - original draft preparation: [Md Ajaharul Hossain], Writing - review and editing: [Md Ajaharul Hossain, Ramakrishna Bandi], Supervision: [Ramakrishna Bandi];
% \subsection*{Funding}
% Not applicable.
\subsection*{Availability of data and materials}
The datasets supporting the conclusions of this article are included within the article.
% \section{Application to Entanglement assisted quantum error-correcting codes}
% \begin{lemma}\cite[Corollary 1]{wilde2008optimal}
%  Let $H_{1}$ and $H_{2}$ be parity check matrices of two linear codes $D_{1}$ and $D_{2}$ with parameters $\left[n, k_{1}, d_{1}\right]_{q}$ and $\left[n, k_{2}, d_{2}\right]_{q}$, respectively. Then an $\left[\left[n, k_{1}+\right.\right.$ $\left.\left.k_{2}-n+c, \min \left\{d_{1}, d_{2}\right\} ; c\right]\right]_{q}$ EAQECC can be obtained where $c=\operatorname{rank}\left(H_{1} H_{2}{ }^{t}\right)$ is the required number of maximally entangled states.
% \end{lemma}

% \begin{lemma}\cite[Proposition 4.2]{Guenda2019}
% \label{elleaqec}
%  Let $\ell \geq 0$ be an integer and $C_{1}$ and $C_{2}$ be a linear $\ell$-intersection pair of codes with parameters $\left[n, k_{1}, d_{1}\right]_{q}$ and $\left[n, k_{2}, d_{2}\right]_{q}$, respectively. Then there exists an $\left[\left[n, k_{2}-\ell, \min \left\{d_{1}^{\perp}, d_{2}\right\} ; k_{1}-\ell\right]\right]_{q}$ EAQECC with $d_{1}^{\perp}=d\left(C_{1}^{\perp}\right)$.
% \end{lemma}

% %%%%%%%%%%%%%%%%%%%%%%%%%%%%%%%%%%%%%%%%%%%%%%%%%%%%%%%%%%%%%%%%%%%%%%%%%%%%%%%%%%%%%%%%%%%%%%%%%%%%%%%%%%%
% \section{Conclusion}

\section*{Acknowledgements}
The first author would like to thank IIIT Naya Raipur for the financial support to carry out this work. The second author is supported by the National Board of Higher Mathematics, Department of Atomic Energy, India through project No. 02011/20/2021/ NBHM(R.P)/R\&D II/8775.

\bibliography{lcc}% common bib file
%% if required, the content of .bbl file can be included here once bbl is generated
%%\input sn-article.bbl

%% Default %%
%%\input sn-sample-bib.tex%

\end{document}